\documentclass{article}

\usepackage{PRIMEarxiv}

\usepackage[utf8]{inputenc} 
\usepackage[T1]{fontenc}    
\usepackage{hyperref}       
\usepackage{url}            
\usepackage{booktabs}       
\usepackage{amsfonts}       
\usepackage{nicefrac}       
\usepackage{microtype}      
\usepackage{lipsum}
\usepackage{fancyhdr}       
\usepackage{graphicx}       
\graphicspath{{media/}}     

\usepackage{amsmath}
\usepackage{amsthm}

\newtheorem{theorem}{Theorem}[section]
\newtheorem{corollary}{Corollary}[theorem]
\newtheorem{lemma}[theorem]{Lemma}
\newtheorem{remark}[theorem]{Remark}
\newtheorem{definition}[theorem]{Definition}

\title{The ultrametric backbone is the union of all minimum spanning forests
}

\author{
  Jordan C. Rozum, Luis M. Rocha \\
  Department of Systems Science and Industrial Engineering \\
  Binghamton University (State University of New York) \\
  Binghamton, NY\\
  \texttt{jrozum@binghamton.edu} \\
}

\begin{document}
\maketitle

\begin{abstract}
Minimum spanning trees and forests are powerful sparsification techniques that remove cycles from weighted graphs to minimize total edge weight while preserving node reachability, with applications in computer science, network science, and graph theory. Despite their utility and ubiquity, they have several limitations, including that they are only defined for undirected networks, they significantly alter dynamics on networks, and they do not generally preserve important network features such as shortest distances, shortest path distribution, and community structure. In contrast, distance backbones, which are subgraphs formed by all edges that obey a generalized triangle inequality, are well defined in directed and undirected graphs and preserve those and other important network features. The backbone of a graph is defined with respect to a specified path-length operator that aggregates weights along a path to define its length, thereby associating a cost to indirect connections. The backbone is the union of all shortest paths between each pair of nodes according to the specified operator. One such operator, the $\max$ function, computes the length of a path as the largest weight of the edges that compose it (a weakest link criterion). It is the only operator that yields an algebraic structure for computing shortest paths that is consistent with De Morgan's laws. Applying this operator yields the ultrametric backbone of a graph in that (semi-triangular) edges whose weights are larger than the length of an indirect path connecting the same nodes (i.e., those that break the generalized triangle inequality based on $\max$ as a path-length operator) are removed. We show that the ultrametric backbone is the union of minimum spanning forests in undirected graphs and provides a new generalization of minimum spanning trees to directed graphs that, unlike minimum equivalent graphs and minimum spanning arborescences, preserves all $\max-\min$ shortest paths and De Morgan’s law consistency.
\end{abstract}


\section{Introduction and background}

Many problems in network science and graph theory, such as predicting links ~\cite{rocha_recommender_2003}, optimizing traversal~\cite{van_mieghem_weight_structure_2005,simas_distance_2021}, identifying primary transmission modes in spreading dynamics~\cite{correia_metric_2023}, locating central (or redundant) nodes and edges~\cite{simas_distance_2021}, defining community structure~\cite{correia_metric_2023}, or predicting the size of cascades 
when nodes or edges are attacked, depend strongly on the structure of shortest paths~\cite{van_mieghem_phase_transition_2005}.
Often the length of a path is computed as the sum of its edge weights but the underlying system or process may suggest other choices, such as multiplying the edge weights or taking only the largest edge weight. 
The method of aggregating edge weights determines the distances between nodes and which paths are shortest in the context of a specific optimization problem. The aggregation operation encodes the \textit{cost} of indirect associations or interactions.
Moreover, other methods beyond shortest paths, such as diffusion and resistance distances, are possible to aggregate indirect associations in networks \cite{estrada2010,silver2018tuned}.
A unifying framework to study families of algebraically consistent edge weighting and path aggregation that quantifies node-to-node distance in weighted graphs is provided by the \textit{distance closure} \cite{simas_distance_2015} (see appendix for details). 

The distance closure applies to distance-weighted graphs $G=(X,D)$ with node set $X$ and edge set $D$. In $G$, $d_{i,j}$ takes values in the extended real numbers $[0,\infty]$ and denotes the weight of an edge $(x_i,x_j)$ from $x_i\in X$ to $x_j\in X$ when it is finite, and the absence of an edge when it is infinite. When it is clear from context, $D=D(X,X)$ may also refer to the weighted adjacency matrix with entries $d_{i,j}$ (where non-edges take on the value $\infty$).
Similarly, when context allows, $d_{i,j}$ may refer to the edge $(x_i,x_j)$ itself (when $d_{i,j}$ is finite) or to the absence of such an edge (when $d_{i,j}=\infty$).
\footnote{We make no distinction between an edge of weight $\infty$ and a non-edge. In particular, this means that all edges have finite weight, and thus a path cannot contain edges with infinite weight. Similarly, the weight of a graph, which refers to the sum of its edges, is the sum over only the finite entries of its adjacency matrix. Removing an edge $(x_i,x_j)$ is equivalent to setting $d_{i,j}=\infty$, and similarly, adding an edge $(x_i,x_j)$ is equivalent to setting $d_{i,j}$ to a finite (non-negative) value. Comparisons or binary operations on edges are to be understood as operating on the edge weights.}

In this framework, applying triangular metric space operations \cite{Menger,Klir1995} leads to general algebraic definitions of network distances including shortest path distance, diffusion distance, and resistance as the closure of an algebraic structure \cite{simas_distance_2015}. 
Specifically, one defines a non-decreasing, commutative, and associative operation $\otimes:[0,\infty]\times [0,\infty]\rightarrow[0,\infty]$ with identity $0$ for accumulating edge weights along a path to determine path length (e.g., the sum of edge weights), and a decreasing, commutative, and associative operation $\oplus:[0,\infty]\times [0,\infty]\rightarrow[0,\infty]$ with identity $\infty$ for aggregating path lengths to determine the distance between nodes (e.g., the minimum of path lengths). These two operations form monoids on $[0,\infty]$, and as a pair $\mathcal{D}=(\oplus,\otimes,[0,\infty])$ they determine how node-to-node distances are computed. The \emph{distance closure} with respect to $\mathcal{D}$ is $G^{\mathcal{D}}=(X,D^{\mathcal{D}})$ with edge weights $d^\mathcal{D}_{i,j}$ given by the distances between nodes as determined by the operators of $\mathcal{D}$ (and infinite between the nodes where no path exists in $G$). 

As a concrete example, the traditional distance closure used in network science is formed using $\oplus\equiv\min$ and $\otimes\equiv+$ (i.e., $\mathcal{D}=(\min,+,[0,\infty])$), so that the length of a path is the sum of its edge weights, and the distance from node $x_i$ to node $x_j$ is the smallest length among all paths from $x_i$ to $x_j$. In this case, the closure $G^\mathcal{D}$ is the traditional graph of node-to-node distances.

In the usual case of shortest-path distances ($\oplus\equiv\min$), the distance between nodes is the length of the shortest path between them for some choice of edge-weight aggregation operator. In this case,
a \textit{generalized triangle inequality} 
provides a transitivity criterion that separates edges into two categories: triangular edges that obey it, and semi-triangular edges that do not.
Triangular edges are necessary and sufficient to compute all shortest paths, while semi-triangular edges are redundant for this purpose. The shortest distance between two nodes connected by a semi-triangular edge is not along the direct path; a shortcut via an indirect path composed of triangular edges must exist.
The subgraph composed of only the triangular edges, the \textit{distance backbone}, thus captures all shortest-path phenomena on the network \cite{simas_distance_2021, correia_metric_2023, costa_directed_2023}.
Moreover, the distance backbone of a weighted (directed or undirected) graph with positive edge weights is the smallest subgraph that preserves all paths of minimal length between every pair of nodes \cite{rocha_recommender_2003,simas_distance_2021}
\footnote{Triangular edges are sufficient to compute all shortest distances, but when there are multiple paths of equal length between a pair of nodes, it is possible to remove some backbone edges or paths and still preserve all shortest distances (but not all shortest paths) in what is known as a transitive reduction \cite{simas_distance_2021}. To preserve all possible paths of minimal length, the entire distance backbone is strictly necessary.}. 

Because there are various possibilities for how to aggregate edge weights to compute a path length, there are various corresponding distance backbones \cite{simas_distance_2015,simas_distance_2021,costa_directed_2023}. 
The most straightforward of these computes path length as the sum of edge weights. 
In the distance closure framework, it is called the \emph{metric} backbone because it derives from the standard triangle inequality of metric spaces \cite{rocha_recommender_2003,correia_metric_2023}: an edge of weight
$d_{i,j}$ between nodes $x_i$ and $x_j$ is removed in the graph's metric backbone if and only if it violates the standard triangle inequality for some intermediary node $x_k$, that is, if and only if $d_{i,j} > d^\mathcal{D}_{i,k}+d^\mathcal{D}_{k,j}$, where $d^\mathcal{D}_{i,k}$ and $d^\mathcal{D}_{k,j}$ are edge weights in the distance closure with respect to $\mathcal{D}=(\min,+,[0,\infty])$ (cf. equation 3.2 of \cite{simas_distance_2021}).

The metric backbone as a sparsification technique has been shown to preserve spreading dynamics and community structure in various application domains including social networks \cite{correia_metric_2023}, brain networks \cite{SimasPhd2012, suckling2015winding,simas2015semi}, protein-protein interaction networks in disease \cite{correia2022conserved}, epidemic spread \cite{correia_metric_2023}, and many others \cite{simas_distance_2021}. 
The backbone generalizes the union of all shortest paths graph (which has the same graph structure as the metric backbone) studied in~\cite{van_mieghem_weight_structure_2005,van_mieghem_phase_transition_2005,van_mieghem_observable_2009} and provides additional algebraic grounding for redundancy analysis in networks.
Indeed, the backbone framework, as discussed above, introduces the concept of semi-triangular edges \cite{rocha_recommender_2003} that is useful in characterizing the redundancy and robustness of shortest paths \cite{simas_distance_2021}. Semi-triangular edges can be ranked by their \textit{distortion}, or the ratio between the direct and shortest paths between the nodes they connect. These edges and their properties are useful for recommendation and link prediction \cite{simas_distance_2015} and can be used to improve epidemic spread predictions \cite{panos_semi-metric_2023}.
Furthermore, by framing shortest-path distances as a result of generalized triangle inequalities in metric spaces, the methodology of distance closures and backbones provides a unifying framework for studying different ways to compute distances on networks, that is, ways to quantify indirect multivariate associations, such as the $\max$ measure of path length, considered next.

Aggregating edge weights along a path using the $\max$ function instead of summation results in the sparser \textit{ultrametric backbone} \cite{simas_distance_2021}. In this case, a path's length is determined solely by the weight of its most costly edge, and the shortest distance from one node to another is, as always, given by the smallest length of all the paths that connect them~\cite{simas_distance_2021,costa_directed_2023}. Thus the ultrametric backbone removes an edge if its endpoints are connected by a path composed of smaller edges.

Following \cite{simas_distance_2015,costa_directed_2023}, we define the ultrametric backbone in Definition \ref{def:UMB}. The general formulation of distance backbones, of which this is a special case, is provided in the appendix.

\begin{definition}[ultrametric backbone]\label{def:UMB}
    The \emph{ultrametric backbone} $U=(X,B^{\max})$ of a distance graph $G=(X,D)$ is the subgraph formed by the edges of $G$ that have invariant weight under ultrametric closure $G^\mathcal{U}$ where $\mathcal{U}=(\min,\max,[0,\infty])$.
\end{definition}

From the definition of the ultrametric closure \cite{simas_distance_2015}, the edges retained in the ultrametric backbone are precisely the edges $(x_i,x_j)$ satisfying $d_{i,j}=d^\mathcal{U}_{i,j}$ where \
    \begin{equation}
    d^\mathcal{U}_{i,j}=\min_{\substack{\pi\text{ is a path from}\\ \text{$x_i$ to $x_j$ in } G_D}}\quad\max_{(x_k,x_l)\in\pi} d_{k,l}.
    \end{equation} 

In other words, the ultrametric backbone $U=(X,B^{\max})$ of $G=(X,D)$ is the subgraph obtained by removing an edge $d_{i,j}$ from $G$ (i.e., setting $b^{\text{max}}_{i,j}=\infty$) if and only if there exists a path $\pi$ in $G$ from $x_i$ to $x_j$ in which every edge of $\pi$ has weight strictly less than $d_{i,j}$.
    
Because the maximum edge weight along a path is never more than the sum of edge weights along that path (assuming positive edge weights), a graph's ultrametric backbone is always a subgraph of its metric backbone \cite{costa_directed_2023}. Compared to the metric backbone, the ultrametric backbone places more emphasis on the highest cost edges in a path. As such, its most natural applications concern processes in which bottlenecks dominate. For example, in package delivery, the maximum number of items that can be delivered in one trip along a given route is determined by the leg of the journey on which the fewest number of items can be taken. In this example, the cost is inversely related to the maximum item capacity. Indeed it is often the case that transforming edge weights between a distance space (cost) and a proximity space (e.g., normalized item capacity) is conceptually and analytically beneficial\footnote{For instance, the minimum distance path length operator and the maximum distance edge weight operator in distance space correspond to the maximum proximity path length operator and minimum proximity edge weight operators in proximity space, respectively. When using the maximum proximity path length operator, the minimum proximity edge weight operator is the only operator consistent with De Morgan's laws, meaning it is the only complementary operator that can be used to form a fuzzy logic\cite{simas_distance_2015}. Thus, in distance space, taking a path's length to be the maximum of its edge weights yields the only algebraic structure for computing shortest paths that admits a consistent notion of logical negation.}. We have formally described this connection in previous work \cite{simas_distance_2015,simas_distance_2021}, which we briefly review in the appendix.

It is important to emphasize that neither the metric backbone nor the ultrametric backbone of a connected undirected graph is equivalent to a minimum spanning tree (MST), which is a minimum-weight subgraph that connects all vertices. This is easily verified by considering the complete graph on three vertices with equal edge weights. In this case, both the metric and ultrametric backbones are equal to the original graph, and hence are not trees. In contrast, any MST of this graph is a two-path. In the directed case, MSTs are not defined, though two prominent generalizations exist: minimum equivalent graphs, and minimum spanning arborescences (see, e.g., \cite{aho_transitive_1972,hsu_algorithm_1975,bang-jensen_minimum_2008,camerini_min-max_1978,korte_spanning_2018}). A minimum equivalent graph is a reachability-preserving subgraph of minimum (summed) weight, and a minimum spanning arborescence (at a root node $x_r$) is a directed acyclic graph rooted at $x_r$ that preserves reachability from $x_r$ to all other nodes (except $x_r$ itself) and is of minimal weight\footnote{We note that here, reachability refers to whether or not a path exists between two nodes and does not consider the length of that path. Thus, a subgraph may preserve reachability but not distances.}. Neither subgraph is equal to the metric or ultrametric backbone, in general, as is easily verified by, again, considering the complete directed graph on three vertices with equal edge weights.

Nevertheless, there are similarities between the concept of a graph backbone and various minimum spanning subgraphs. The results presented here formalizes the connection between the ultrametric backbone and MSTs in undirected graphs. Further, we demonstrate that the ultrametric backbone provides a distinct approach to generalizing the concept of an MST to directed graphs.

\section{Results}
\subsection{Undirected graphs}
The main result of this section is that the ultrametric backbone of a positively edge-weighted undirected graph is the union of all minimal spanning forests. This result follows immediately from the special case of connected graphs, which is formalized in Theorem \ref{res:U_is_MST_union}. This theorem relies on a few Lemmas, the first of which is a well-known property of MSTs.
\begin{lemma}[Cycle Property of MSTs]\label{res:cycle_property}
    Let $G$ be a connected, distance-weighted, undirected graph. For any cycle $\sigma$ in $G$, if $\sigma$ has an edge $d_{i,j}$ larger than all other edges in $\sigma$, then $d_{i,j}$ is not part of any MST of $G$.
\end{lemma}
Lemma \ref{res:cycle_property} is a standard result in introductory graph theory. A simple proof by contradiction is given here.
\begin{proof}[Proof of Lemma \ref{res:cycle_property}]
    Assume, by way of contradiction, that a cycle $\sigma$ of edges in $G=(X,D)$ has an edge $d_{i,j}$ larger than the weight of any other edge in $\sigma$ and that $T=(X,D_T)$ is an MST of $G$ in which the edge $(d_T)_{i,j}$ of $T$ corresponds to an edge $d_{i,j}$. Removing $d_{i,j}$ from $T$ yields two disconnected trees covering the node set $X$: $T_1=(X_1,D_1)$ and $T_2=(X_2,D_2)$ with $x_i$ in $X_1$ and $x_j$ in $X_2$. The path obtained from the cycle $\sigma$ by removing the edge $d_{i,j}$ connects $x_i$ and $x_j$ in $G$. Because this path connects $x_i\in X_1$ and $x_j\in X_2$, it must contain an edge $d_{k.l}$ that is strictly smaller than $d_{i,j}$ (by the contradiction assumption) with $x_k$ in $X_1$ and $x_l$ in $X_2$. Therefore, the graph $T^\prime$ on $X$ with edges $(d_{T^\prime})_{i,j}=\min \{(d_1)_{i,j},(d_2)_{i,j}\}$ for $(i,j)\neq (k,l)$ and $(d_{T^\prime})_{k,l}=d_{k,l}$ is a spanning tree of $G$ with weight strictly less than the weight of $T$, contradicting the assumption that $T$ is an MST.
\end{proof}

From this result and elementary facts about the ultrametric backbone, it follows that every MST is a subgraph of the ultrametric backbone, as stated in Lemma \ref{res:U_has_all_T}.

\begin{lemma}\label{res:U_has_all_T}
    Let $G=(X,D)$ be a connected, weighted, undirected graph with positive edge weights and ultrametric backbone $U=(X,B^{\max})$. Any MST $T$ is a subgraph of $U$.
\end{lemma}

\begin{proof}[Proof of Lemma \ref{res:U_has_all_T}]
    It suffices to show that every edge removed in $U$ is also removed in any MST. 
    Consider $(i,j)$ such that $d_{i,j}<\infty$ is an edge (of $G$), but $b^{\max}_{i,j}=\infty$ is not an edge of $U$ (in the case $G=U$, the result holds trivially). 
    Because $b^{\max}_{i,j}=\infty$, it follows immediately from Definition \ref{def:UMB} that $d_{i,j}$  belongs to a cycle in $G$ whose other edges are strictly smaller. 
    By Lemma \ref{res:cycle_property} this edge does not belong to any MST.
\end{proof}

The counterpart to Lemma \ref{res:U_has_all_T} is Lemma \ref{res:no_extra_edges_in_U}, which states that every edge of the ultrametric backbone belongs to at least one MST. Essentially, the result follows from two facts. The first fact is that for any edge $d_{i,j}$ belonging to the ultrametric backbone and any path $\pi$ connecting $x_i$ to $x_j$ that does not contain $d_{i,j}$, there is at least one edge $d_{k,l}$ in $\pi$ with $d_{i,j}\leq d_{k,l}$. The second fact is is that $\pi$ and $d_{i,j}$ together form a cycle, so replacing $d_{k,l}$ with $d_{i,j}$ results in an alternate path connecting the nodes in $\pi$ (with an equal or lesser edge weight sum). Such an operation cannot introduce cycles if $\pi$ is contained in an MST because removal of any edge participating in the introduced cycles would produce a strictly lower-weight connected graph. The formalization of this reasoning gives rise to the proof of Lemma \ref{res:no_extra_edges_in_U}.

\begin{lemma}\label{res:no_extra_edges_in_U}
    Let $G=(X,D)$ be a connected, weighted, undirected graph with positive edge weights and ultrametric backbone $U=(X,B^{\max})$. For any edge $b^{\max}_{i,j}$ of the ultrametric backbone, there exists an MST of $G$ containing the corresponding edge $d_{i,j}$.
\end{lemma}
\begin{proof}[Proof of Lemma \ref{res:no_extra_edges_in_U}]
    Consider an arbitrary MST $T=(X,D_T)$ of $G$ and an arbitrary edge $b^{\max}_{i,j}<\infty$ of the ultrametric backbone $U$ of $G$. If $(d_T)_{i,j}$ is an edge of $T$, then the Lemma holds in this case. Otherwise, we seek to construct an alternate MST $T^\prime$ that contains $(d_{T^\prime})_{i,j}$ as an edge. In this case, $x_i$ must be connected to $x_j$ in $T$ via a path $\pi$ that does not contain $d_{i,j}$. Because $b^{\max}_{i,j}$ is an edge of $U$, $\pi$ must not be composed of edges whose weights are all strictly smaller than $d_{i,j}$. That is, $\pi$ contains an edge $d_{k,l}$ with $d_{k,l}\geq d_{i,j}$. We construct the graph $T^\prime\equiv (X, D_{T^\prime})$ from $T$ by removing the edge between $x_k$ and $x_l$ and inserting the edge between $x_i$ and $x_j$. The altered graph $T^\prime$ remains connected because the path $\pi^\prime$ formed by replacing $d_{k,l}$ by $d_{i,j}$ in $\pi$ connects $x_k$ to $x_l$ in $T^\prime$. Furthermore, because $d_{k,l}\geq d_{i,j}$, the weight of $T^\prime$ is not larger than that of $T$. From this fact and the connectedness of $T^\prime$, it follows that $T^\prime$ must be a tree of weight equal to the weight of $T$ (i.e., $d_{k,l}=d_{i,j}$); otherwise $T^\prime$ provides a counterexample to the assumption that $T$ is an MST. Therefore, either $T$ or $T^\prime$ is an MST of $G$ with edge $d_{i,j}$.
\end{proof}

The main result of this section, which follows immediately from Lemmas \ref{res:U_has_all_T} and \ref{res:no_extra_edges_in_U}, is that the ultrametric backbone of a connected undirected graph with positive edge weights is the exactly equal to the (non-disjoint) union of that graph's MSTs. This result is formally stated as Theorem \ref{res:U_is_MST_union}.

\begin{theorem}\label{res:U_is_MST_union}
    Let $G=(X,D)$ be a connected, weighted, undirected graph with positive edge weights and ultrametric backbone $U=(X,B^{\max})$, and let $\mathcal{T}$ be the set of all MSTs of $G$. Then $U=\bigcup_{T\in \mathcal{T}}T$ is the (non-disjoint) graph union of all MSTs.
\end{theorem}

Theorem \ref{res:U_is_MST_union} extends in a straightforward way to unconnected graphs by considering minimum spanning forests (a minimum spanning forest of a graph $G$ consists of one MST from each component of $G$). Applying Theorem \ref{res:U_is_MST_union} to each component of a disconnected graph immediately gives rise to Corollary \ref{res:U_is_MSF_union}.
\begin{corollary}\label{res:U_is_MSF_union}
    Let $G$ be a weighted undirected graph with positive edge weights and ultrametric backbone $U$, and let $\mathcal{F}$ be the set of all minimum spanning forests of $G$. Then $U=\bigcup_{F\in \mathcal{F}}F$ is the (non-disjoint) graph union of all minimum spanning forests.
\end{corollary}

Corollary \ref{res:U_is_MSF_union} is analogous to the fact that the metric backbone is the union of all shortest path trees as described in the introduction~\cite{van_mieghem_weight_structure_2005,simas_distance_2021}.

\subsection{Directed graphs}

This section demonstrates, by way of counterexample, that Theorem \ref{res:U_is_MST_union} does not generalize to the case of directed graphs, suggesting that the directed ultrametric backbone extends the concept of MSTs to directed graphs in a manner distinct from traditional constructions.

\begin{remark}\label{res:U_is_not_MSSS_union}
    There exists a weighted directed graph $G$ with positive edge weights and ultrametric backbone $U$ such that the union $\mathcal{S}$ of all minimum equivalent graphs of $G$ satisfies i) $\mathcal{S}$ lacks an edge belonging to $U$ and ii) $\mathcal{S}$ has an edge not in $U$.
\end{remark}
\begin{remark}\label{res:U_is_not_MSA_union}
    There exists a weighted directed graph $G$ with positive edge weights and ultrametric backbone $U$ such that the union $\mathcal{A}$ of all minimum spanning arborescences of $G$ satisfies i) that $\mathcal{A}$ lacks an edge belonging to $U$ and ii) that $\mathcal{A}$ has an edge not in $U$.
\end{remark}

These remarks are demonstrated in Figure \ref{fig:fig1}. The original graph, $G$, is depicted in Figure \ref{fig:fig1}a. The ultrametric backbone, $U$, of this graph is depicted in Figure \ref{fig:fig1}b. In this example, $G$ has a unique minimum equivalent graph, depicted in Figure \ref{fig:fig1}c. The edge $d_{2,3}=5$ is required to preserve $\max-\min$ shortest paths and thus is included in the ultrametric backbone. This edge, however, is absent in the minimum equivalent graph. On the other hand, the edge $d_{2,4}=6$ is present in the minimum equivalent graph but is redundant for $\max-\min$ shortest paths and therefore not present in the ultrametric backbone. Thus, the relationship between these two graphs is as described in Remark \ref{res:U_is_not_MSSS_union}. This same example network demonstrates the correctness of Remark \ref{res:U_is_not_MSA_union} as well, as shown in Figure \ref{fig:fig1}d, which depicts the minimum spanning arborescences rooted at each node. The edge $d_{2,3}$ is absent in all minimum spanning arborescences, but is required in the ultrametric backbone to preserve $\max-\min$ shortest paths; on the other hand, the edge $d_{2,4}$ is not in the ultrametric backbone (it is redundant for $\max-\min$ shortest paths), but it is present in two of the minimum spanning arborescences.

\begin{figure}
  \centering
  \includegraphics[width=1\textwidth]{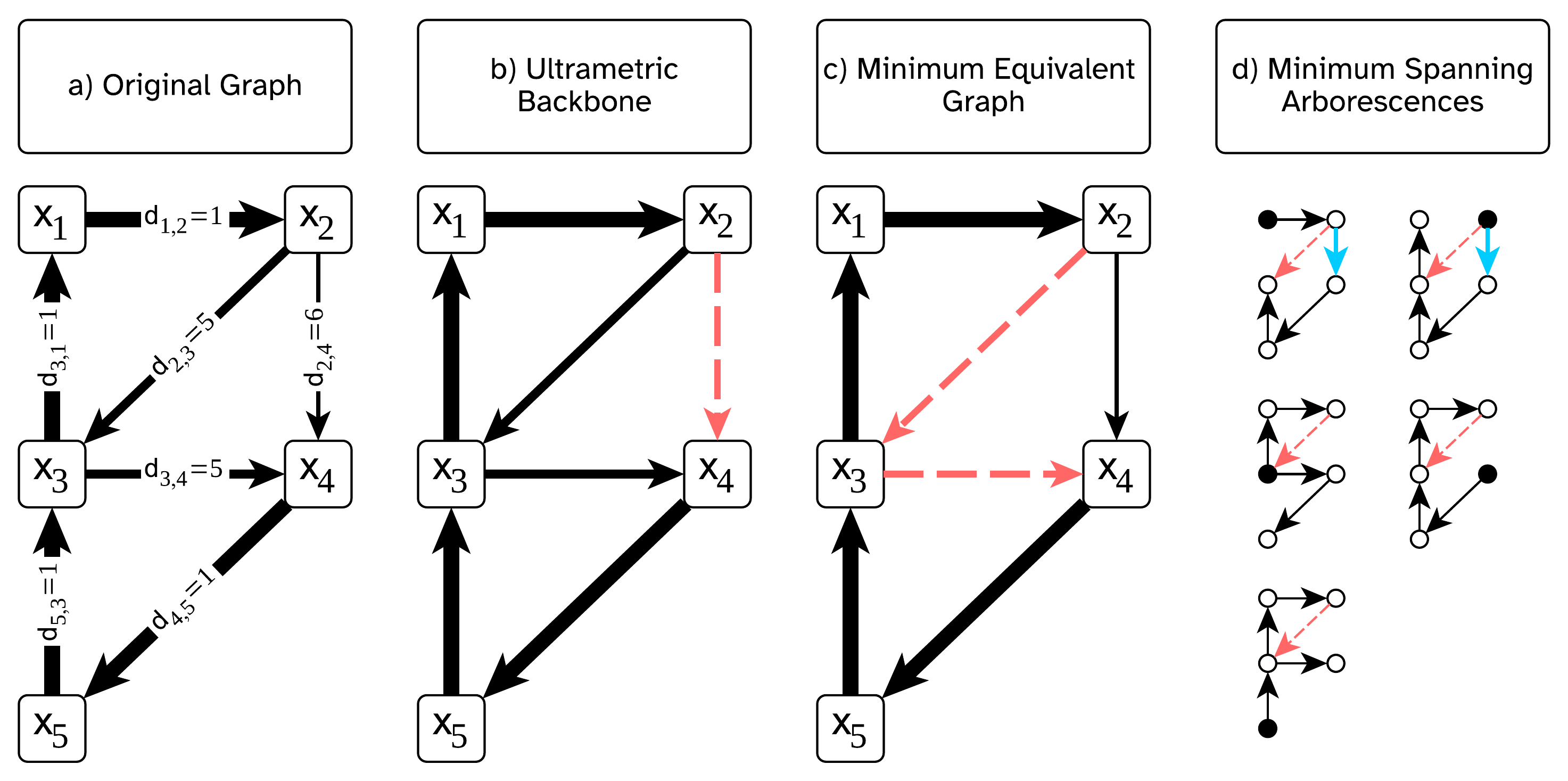}
  \caption{The ultrametric backbone is distinct from unions of MST analogs in directed graphs. (a) An example distance graph with thicker edges corresponding to smaller distance weights. (b) The ultrametric backbone is shown with edge weights omitted for visual clarity. Edge $d_{2,4}$ is removed, as indicated by the red dashes; it is redundant for $\max-\min$ shortest paths because it breaks the $\max-\min$ transitivity. (c) A minimum equivalent graph is shown, which in this example is unique. Note that it is distinct from the ultrametric backbone and does not preserve the shortest $\max-\min$ path from $x_2$ to $x_4$. (d) Five (in this case, unique) minimum spanning arborescences with the root node filled in with black are shown. The red dashed line indicates an edge, $d_{2,3}$, that is not in any minimum spanning arborescence, but is in the ultrametric backbone and required for $\max-\min$ shortest paths (its weight increases from $5$ to $6$). The blue edge, $d_{2,4}$, is present in the union of these five graphs, but is redundant for $\max-\min$ shortest paths and therefore is not in the ultrametric backbone.}
  \label{fig:fig1}
\end{figure}

A crucial property of the ultrametric backbone that gives rise to Lemma \ref{res:U_has_all_T} in the undirected case is that every edge removed in the undirected ultrametric backbone belongs to a cycle. Notably, this property does not hold in directed graphs, where the removed edges need not participate in a cycle. Rather, all that is required is that an alternate path exists between the parent and child nodes of the removed directed edge. The failure of this property to generalize, however, is not sufficient to explain the counterexample of Figure \ref{fig:fig1}: The edge $d_{2,3}$ participates in a cycle in which its weight is strictly larger than that of all others (i.e., $x_1\rightarrow x_2 \rightarrow x_3$). Rather, because $d_{2,3}$ is maximal in the cycle $x_1\rightarrow x_2 \rightarrow x_3$, this counterexample illustrates the failure of Lemma \ref{res:cycle_property} to generalize to minimum equivalent graphs and minimum spanning arborescences.

We note that the metric backbone is equivalent to the union of of all minimum spanning arborescences because it can be constructed as the union of all shortest paths \cite{costa_directed_2023}.

\section{Discussion}
The main result of this work, Theorem \ref{res:U_is_MST_union} and its corollary, is that the ultrametric backbone of any positively weighted undirected graph is the union of all MSTs (or forests, if the graph is not connected). This is surprising because the weight of a spanning tree is defined to be the sum of its edges, but the ultrametric backbone can be defined and computed without any summation. The result is all the more surprising because it does not generalize to natural analogs of MSTs for directed graphs. This suggests that the ultrametric backbone provides a new way to extend the concept of an MST to directed graphs. 

The ultrametric backbone may be especially useful when considering MSTs of graphs in which edge weights have relatively high uncertainty. By binning or coarse-graining edge weights and computing the ultrametric backbone, the ``true'' MST is guaranteed to be retained as a subgraph. Furthermore, potentially relevant edges that are marginally excluded in an MST are not forcibly discarded in a coarse-grained ultrametric backbone approach. This is because one is not forced to differentiate between edges with statistically equal weights when computing the ultrametric backbone.

The correspondence between the ultrametric backbone and the MSTs suggests an approach to finding MSTs that may offer computational advantages when edge summation is numerically difficult, for example when edge weights span many orders of magnitude or when differences between edge weights are small but significant. In particular, any MST of a graph is also an MST of that graph's ultrametric backbone, which may be computed without summing edges, using edge ranks only. Indeed, this property of MSTs is exploited in Kruskal's algorithm for finding minimum spanning forests~\cite{kruskal_shortest_1956}. In the special case when the MST is unique (for example, if all edge weights are distinct), it is equal to the ultrametric backbone. 

Theorem \ref{res:U_is_MST_union} emphasizes the similarities between the ultrametric backbone and MSTs in undirected weighted graphs. By extension, this comparison underscores the differences between other distance backbones (most notably the metric backbone) and MSTs. The metric backbone of an undirected graph necessarily contains its ultrametric backbone as a subgraph, which in turn contains all minimum spanning forests. Thus, the metric backbone is a subgraph that contains all MSTs, and which may, in general, contain additional edges in order to preserve all geodesics. In directed graphs, this difference becomes even more dramatically evident.

In directed graphs, the ultrametric backbone provides a natural extension of MSTs that is distinct from minimum equivalent graphs, minimum spanning arborescences, and their unions. The ultrametric backbone is computationally simple to compute and has the advantage of avoiding technical difficulties surrounding reachability or requirements of strong connectedness that other generalizations must contend with. Intuitively, it generalizes MSTs to directed graphs in a manner that emphasizes the removal of weakest links that occurs in the undirected case. The results presented here suggest that the ultrametric backbone (of a directed or undirected graph) may serve as an alternative or supplement to various minimum spanning subgraphs when analyzing network structure. This conclusion motivates further study regarding the dynamical properties of the ultrametric backbone in various contexts.

\section*{Acknowledgments}
This was was funded by the NIH National Library of Medicine Program grant 01LM011945-01 and the Fundação para
a Ciência e a Tecnologia grant 2022.09122.PTDC (https:
//doi.org/10.54499/2022.09122.PTDC) to LMR.

\section{Conflicts of Interest}
The authors declare no conflicts of interest.

\appendix
\section{General distance closure and backbone}\label{sec:GeneralFramework}
This section provides a summary of the distance closure framework of \cite{simas_distance_2015} and the distance backbone framework of \cite{costa_directed_2023}.

We begin with preliminary definitions from fuzzy logic, which motivated the work of \cite{simas_distance_2015}.

\begin{definition}[T-norm]
A \emph{triangular norm} (abbreviated T-norm) $\wedge:[0,1]\times[0,1]\rightarrow[0,1]$ is an associative, commutative, and non-decreasing binary operation on the closed unit interval that has identity $1$, i.e., $x\wedge 1 = x$ for all $x\in[0,1]$.
\end{definition}
\begin{definition}[T-conorm]
A \emph{triangular conorm} (abbreviated T-conorm) $\vee:[0,1]\times[0,1]\rightarrow[0,1]$ is an associative, commutative, and non-decreasing binary operation on the closed unit interval that has identity $0$, i.e., $x\vee 0 = x$ for all $x\in[0,1]$.
\end{definition}

If, for a T-norm $\wedge$ and T-conorm $\vee$, there exists a bijective involution $\neg$ on $[0,1]$ that satisfies the De Morgan property $a\vee b = \neg(\neg a\wedge\neg b)$ for all $a,b\in [0,1]$, then $\wedge$ and $\vee$ are dual and form a fuzzy logic.

Note that every T-norm and T-conorm (whether dual or not) form a monoid on $[0,1]$ (i.e., $(\wedge,[0,1])$ and $(\vee,[0,1])$ are monoids). Pairs of monoids on the same underlying set, and isomorphisms between these pairs, play an important role in \cite{simas_distance_2015}, so we define them carefully here:

\begin{definition}[monoid pair]
    A \emph{monoid pair} is an ordered pair of monoids that share the same underlying set, written $(*,+,M)$. An \emph{isomorphism of monoid pairs} from $(*,+,M)$ to $(\cdot,\star,N)$ is a function $\varphi:M\rightarrow N$ that is an isomorphism from $(*,M)$ to $(\cdot,N)$ and from $(+,M)$ to $(\star,N)$.
\end{definition}

Note that a monoid pair is quite general, as the two operations need not share anything in common other than their underlying set.

The special case of a monoid pair formed from a T-norm and T-conorm is of particular interest.
\begin{definition}[proximity structure]
    A monoid pair $(\wedge,\vee,M)$ with first operation ($\wedge$) a T-norm, second operation $\vee$ a T-conorm, and underlying set $M=[0,1]$ is called a \emph{proximity structure}. (This is called algebraic structure $I$ in \cite{simas_distance_2015}.)
\end{definition}

We now define analogous structures for distance spaces, rather than for proximity spaces.
\begin{definition}[TD-norm]
A \emph{triangular distance norm} (abbreviated TD-norm) $\oplus:[0,\infty]\times[0,\infty]\rightarrow [0,\infty]$ is an associative, commutative, and non-decreasing binary operation on the closed unit interval that has identity $0$, i.e., $x\oplus 0 = x$ for all $x\in[0,\infty]$.
\end{definition}
\begin{definition}[TD-conorm]
A \emph{triangular distance conorm} (abbreviated TD-conorm) $\otimes:[0,\infty]\times[0,\infty]\rightarrow[0,\infty]$ is an associative, commutative, and non-decreasing binary operation on the closed unit interval that has identity $\infty$, i.e., $x\otimes \infty = x$ for all $x\in[0,\infty]$.
\end{definition}
Note that in \cite{simas_distance_2015}, the functions $f$ and $g$ are introduced to define TD-norms and TD-conorms. In our notation, $f(a,b)=a\oplus b$ and $g(a,b)=a\otimes b$.
\begin{definition}[distance structure]
    A monoid pair $(\oplus,\otimes,M)$ with first operation $\oplus$ a TD-norm, second operation $\otimes$ a TD-conorm, and underlying set $M=[0,\infty]$ is called a \emph{distance structure}. (This is called algebraic structure $II$ in \cite{simas_distance_2015}.)
\end{definition}

In \cite{simas_distance_2015}, it is highlighted that for any proximity structure $\mathcal{P}$, any monotonically decreasing bijection $\varphi:[0,1]\rightarrow[0,\infty]$ can be interpreted as an isomorphism of monoid pairs, thereby inducing a corresponding distance structure $\mathcal{D}$. This allows for convenient extension of earlier results on proximity spaces and fuzzy logic relations (e.g., from \cite{Klir1995}) to distance graphs. In particular, \cite{simas_distance_2015} focuses on the construction of graph closures.

A proximity graph $G_P=(X,P)$ is a weighted graph with edge weights $p_{i,j}$ taken from $[0,1]$, while a distance graph $G_D$ is a weighted graph with edge weights $d_{i,j}$ taken from $[0,\infty]$. The proximity closure $G_P^{\mathcal{P}}$ of a proximity graph of $G_P$ with respect to a proximity structure $\mathcal{P}=(\wedge,\vee,[0,1]$ has the same node set and is a complete graph. The edge weight of the edge $(x_i,x_j)$ in $G_P^\mathcal{P}$ is denoted $p^\mathcal{P}_{i,j}$ given by 

\begin{equation}
    p^\mathcal{P}_{i,j}=\bigwedge_{\substack{\pi\text{ is a path from}\\ \text{$x_i$ to $x_j$ in } G_P}}\quad\bigvee_{(x_k,x_l)\in\pi} p_{k,l}.
\end{equation}
If no path exists between $x_i$ and $x_j$ in $G_P$, the weight of the corresponding edge is $0$ in $G_P^{\mathcal{P}}$.

By application of any monotonically decreasing bijection $\varphi:[0,1]\rightarrow[0,\infty]$, and imposition of the homomorphism property, an analogous structure, the distance closure $G_D^{\mathcal{D}}$ of the distance graph $G_D$ with respect to $\mathcal{D}=(\otimes,\oplus,[0,\infty])$ corresponding to $G_\mathcal{P}$ under $\varphi$ can be constructed. The edge weight of the edge $(x_i,x_j)$ in $G_D^{\mathcal{D}}$ is given by 

\begin{equation}
    d^\mathcal{D}_{i,j}=\bigoplus_{\substack{\pi\text{ is a path from}\\ \text{$x_i$ to $x_j$ in } G_D}}\quad\bigotimes_{(x_k,x_l)\in\pi} d_{k,l}.
\end{equation}
If no path exists between $x_i$ and $x_j$ in $G_\mathcal{D}$, the weight of the corresponding edge is $\infty$ in $G_D^{\mathcal{D}}$. 

As proved in \cite{simas_distance_2015} for the case of undirected graphs and in \cite{costa_directed_2023} for the case of directed graphs, the computation of the closure graph commutes with application of $\varphi$ to the edge weights.

In this framework, the distance backbone of a distance graph $G_D$ with respect to $\mathcal{D}$ is defined for $a\oplus b = \min(a,b)$ according to \cite{simas_distance_2021,correia_metric_2023,costa_directed_2023}. 
\begin{definition}[distance backbone]
    The \emph{distance backbone} of $G_D=(X,D)$ with respect to a distance structure $\mathcal{D}=(\min,\otimes,[0,\infty])$ is the subgraph $(X,B^{\otimes})$ formed by the edges of $G_D$ that have invariant weight under closure, i.e., edges $(x_i,x_j)$ satisfying $d_{i,j}=d^\mathcal{D}_{i,j}$.
\end{definition}
For example, in the case of the metric backbone, $a\oplus b=\min(a,b)$ and $a\otimes b=a+b$, while in the case of the ultrametric backbone, $a\oplus b=\min(a,b)$ and $a\otimes b=\max(a,b)$.

\section{Summary of notation used in main text}
In this appendix, we present Table \ref{tab:notation}, which summarizes the mathematical notation used in the main text.
\renewcommand{\arraystretch}{1.5}
\begin{table}[]
    \centering
    \begin{tabular}{c|p{12cm}}
    Symbol & Description \\\hline
    $G=(X,D)$ & A distance graph with positive edge weights, which may be directed or not, depending on context.\\
    $X$ & The vertex set of a graph.\\
    $x_i$ & A vertex in a distance graph.\\
    $D$ & The edges of a distance graph. When clear from context, $D=D(X,X)$ refers to the adjacency matrix of the graph, with non-edges represented by infinite entries.\\
    $d_{i,j}$ & An edge of a distance graph represented by its indexed value in the adjacency matrix. Writing $d_{i,j}=\infty$ is equivalent to noting that the graph does not have an edge between $x_i$ and $x_j$ (or from $x_i$ to $x_j$ in the directed case).\\
    $\pi$ & Indicates a path in a graph.\\
    $T$, $T^\prime$, etc. & Distance graphs; the letter $T$ is used to indicate that the graph has been shown (or will be shown) to be a tree.\\
    $\mathcal{T}$ & The set of all minimum spanning trees of a distance graph.\\
    $F$ & A distance graph; the letter $F$ is used to indicate that the graph is a minimum spanning forest.\\
    $\mathcal{F}$ & The set of all minimum spanning forests of a distance graph.\\
    $\mathcal{S}$ & The union of all minimum equivalent graphs of a directed distance graph.\\
    $\mathcal{A}$ & The union of all minimum spanning arboresceences (at all roots) of a directed distance graph.\\
    $\mathcal{D}$ & A distance structure; an algebraic structure that determines how distances between nodes in a distance graph are computed. It consists two monoid operations, $\oplus$ and $\otimes$, on the extended real numbers (see below).\\
    $\oplus$ & The TD-norm operator used to aggregates edge weights to compute a path length. $\oplus:[0,\infty]\times[0,\infty]\rightarrow[0,\infty]$ must be associative, commutative, and non-decreasing with identity element $0$.\\
    $\otimes$ & The TD-conorm operator used to aggregate path lengths to compute node-to-node distance. Here, set as $\otimes\equiv\min$ except when noted otherwise. $\otimes:[0,\infty]\times[0,\infty]\rightarrow[0,\infty]$ must be associative, commutative, and non-decreasing with identity element $\infty$.\\
    $\mathcal{U}$ & The ultrametric distance structure with $\oplus=\max$ and $\otimes=\min$ computes the distance between nodes as the length of the shortest path (or paths), where the length of a path is given by the weight of its largest edge.\\
    $G^\mathcal{D}$ & The distance closure of the graph $G$ with respect to the distance structure $\mathcal{D}$. Every (strongly) connected component of $G$ is a complete graph in $G^\mathcal{D}$ with an edge weight $d^\mathcal{D}_{i,j}$ equal to the distance between $x_i$ and $x_j$ in $G$ computed using $\mathcal{D}$.\\
    $U=(X,B^{\max})$ & The ultrametric backbone of a graph $G$. This graph consists of exactly those edges $d_{i,j}$ whose weight is conserved in the distance closure constructed using the distance structure $\mathcal{U}$ (see above). That is, its edges are shortest paths in $G$ where a path's length equals the weight of its largest edge. Edges are denoted using $b^{\max}_{i,j}$.\\\hline
    \end{tabular}
    \caption{Key notation used in the main text. Additional notation is used and defined in Appendix \ref{sec:GeneralFramework}.}
    \label{tab:notation}
\end{table}

\bibliographystyle{unsrt}  
\bibliography{references}

\end{document}